\documentclass[pra,superscriptaddress,twocolumn]{revtex4}

\usepackage{microtype}
\usepackage{amsmath,amssymb,amsfonts, amsthm}
\usepackage{tabularx}
\newcolumntype{Y}{>{\centering\arraybackslash}X}
\usepackage{graphicx}
\usepackage{color}
\usepackage{braket}
\usepackage{bbold}

\newtheorem{thm}{Theorem}
\newtheorem{lm}{Lemma}
\newtheorem{defin}{Definition}
\newtheorem{cor}{Corollary}
 \newcommand{\id}{\mathbb{1}}
\newcommand{\DD}{\mu}
\newcommand{\CD}{m}
\newcommand{\dd}{d}
\newcommand{\mo}{$\CD^{\star}$-orthogonal}
\newcommand{\moity}{$\CD^{\star}$-orthogonality}

\begin{document}
\title{Qudit Color Codes and Gauge Color Codes in All Spatial Dimensions} 
\author{Fern H.E. Watson}
\email{fern.watson10@imperial.ac.uk}
\affiliation{Department of Physics and Astronomy, University College London, Gower Street, London WC1E 6BT, UK.}
\affiliation{Department of Physics, Imperial College London, Prince Consort Road, London SW7 2AZ, UK.}
\author{Earl T. Campbell}
\affiliation{Department of Physics and Astronomy, University of Sheffield, Sheffield, S3 7RH, UK.}
\author{Hussain Anwar}
\affiliation{Department of Mathematical Sciences, Brunel University, Uxbridge, Middlesex UB8 3PH, UK.}
\author{Dan E. Browne}
\affiliation{Department of Physics and Astronomy, University College London, Gower Street, London WC1E 6BT, UK.}

\begin{abstract}
Two-level quantum systems, qubits, are not the only basis for quantum computation. Advantages exist in using qudits, $\dd$-level quantum systems, as the basic carrier of quantum information. We show that color codes---a class of topological quantum codes with remarkable transversality properties---can be generalized to the qudit paradigm. In recent developments it was found that in three spatial dimensions a qubit color code can support a transversal non-Clifford gate, and that in higher spatial dimensions additional non-Clifford gates can be found, saturating Bravyi and K\"onig's bound [Phys. Rev. Lett. \textbf{110}, 170503 (2013)]. Furthermore, by using gauge fixing techniques, an effective set of Clifford gates can be achieved, removing the need for state distillation. We show that the qudit color code can support the qudit analogues of these gates, and show that in higher spatial dimensions a color code can support a phase gate from higher levels of the Clifford hierarchy which can be proven to saturate Bravyi and K\"onig's bound in all but a finite number of special cases. The methodology used is a generalisation of Bravyi and Haah's method of triorthogonal matrices [Phys. Rev. A \textbf{86} 052329 (2012)], which may be of independent interest. For completeness, we show explicitly that the qudit color codes  generalize  to gauge color codes, and share the many of the favorable properties of their qubit counterparts.  
\end{abstract}

\pacs{03.67.Pp, 03.67.Lx}

\maketitle

\section{Introduction} 

Quantum technologies are often developed in the qubit paradigm, where the basic carrier of quantum information is a two-level quantum system.  Qubits are a natural choice because binary is the language of classical technologies. However, even here, despite the prevalence of binary, its supremacy is questionable. Indeed, Donald Knuth has advocated the use of balanced ternary, a 3-state classical logic~\cite{Knuth1997}.  In the quantum domain, qudits offer a state space with a richer structure than their two-level counterparts, and the merits of this for quantum information have been explored in many contexts~\cite{Gottesman99,Zhou03,Howard12,Anwar12,Campbell12,Duclos13,Anwar14,Watson14,Hutter14}. Recent experiments have shown even large $\dd$ quantum systems can be precisely controlled~\cite{Smith13,AndersonSosa14}.

Here we present a qudit generalisation of a powerful class of quantum error correcting codes, the color codes~\cite{Bombin06,Bombin07,Bombin13,Bombin13b,Kubica14}.  Along with surface codes~\cite{Kitaev03,Bravyi98}, they constitute the most successful topological codes.  Using topology, quantum codes have achieved high error thresholds, whilst proving more practical than concatenated codes.  Thresholds of qudit surface codes indicate improvements with qudit dimension~\cite{Duclos13,Anwar14,Watson14}.  Qubit color codes have several advantages over qubit surface codes, and we show these features can be transferred over into the qudit setting.

There has been some prior investigation into qudit color codes for prime-power $\dd$~\cite{Sarvepalli10} and color codes based on more general groups~\cite{Brell14}, although these works were restricted to 2D topologies.  In this paper, we generalize color codes to any qudit dimension $\dd$ and spatial dimensions $\DD$ up to the point where $ \DD $ factorial is a multiple of $\dd$.  Specifically, given any lattice suitable for constructing qubit color codes, we show how to use the same lattice to construct a qudit color code.   For qubits, a non-Clifford gate can be implemented in color codes in 3 and higher spatial dimensions transversally, i.e. by a tensor product of local unitary gates, an inherently fault-tolerant procedure~\cite{Knill96,Knill98,Raussendorf06}. To avoid confusion between the spatial dimension of the lattice and the Hilbert space dimension of the qudit, we shall always denote the former by the letter  $\DD$ and the latter by the letter $\dd$.

Recently, it has been shown by Bravyi and K\"{o}nig~\cite{BravKoenig13} that a quantum error correcting code in $\DD$ spatial dimensions can support a gate with constant depth from at most the $\DD$th level of the Clifford hierarchy. The fact that color codes can be shown to saturate this bound with transversal gates is a very promising feature, and when combined with gauge fixing techniques~\cite{Bombin13,Paetznick13,Anderson14,OConnor14} enables universal quantum computation without the need of magic state distillation~\cite{BraKit05,Rei03a,Anwar12,Campbell12,BravHaah12,Jones13,Meier13,Campbell14}.  The structure of the qudit Clifford group is very different from its qubit counterpart~\cite{Appleby05,Howard12}.  Nevertheless, we find that 3D color codes also provide transversal non-Clifford gates in the qudit case.

Recently the color codes were generalized to gauge color codes~\cite{Bombin13}---subsystem codes with many advantageous features---including low weight error detection measurements, universal transversal gates via gauge fixing for $\DD>2$, fault-tolerant conversion~\cite{Anderson14} between codes of different spatial dimension~\cite{Bombin14dimension}, and for the $\DD=3$ case, single shot error correction~\cite{Bombin14singleshot}---a robustness to measurement errors without the need for repeated measurements. We show that the qudit color codes introduced here can also be generalized to gauge color codes.

The main technique which we employ is a bipartition of the vertices in the graph that defines the code into \textit{starred} and \textit{unstarred} vertices. We call this the \textit{star-bipartition}, to distinguish it from the other important colorings which define the color codes. The commutation properties of the stabilizer and logical operators of the color codes (and gauge color codes) in the qubit setting can be reduced to the fact that the pairs of  operators  $X\otimes X$ and  $Z \otimes Z$ commute. The star-bipartition we introduce replaces some operators with their complex conjugate with respect to the computational basis of $Z$ eigenstates. For example, replacing the above operators with $X\otimes X^*$ and $Z \otimes Z^*$, respectively. Crucially, this latter pair of operators does commute for qudits of any dimension, and this becomes the starting point for a generalisation of the  color codes from qubits to qudits of any Hilbert space dimension $\dd$ and any spatial dimension $\DD$. The sharp-eyed reader will note that $X^*=X$. We write the star explicitly on a real matrix here, and throughout, to emphasize symmetry and to simplify notation. Thus we see that as long as a pair of $X$- and $Z$-type stabilizers have in common an equal number of starred qudits, they will commute.

Furthermore, the star-bipartition  provides a general framework for constructing transversal gates from higher levels of the Clifford hierarchy. While elements of this technique can be seen in earlier work~\cite{Sarvepalli10,Bombin13}, this is the first time that it has been exploited systematically. The second key technical component of our work is a generalisation of the  triorthogonal matrix technique by Bravyi and Haah~\cite{BravHaah12}.


This paper is structured as follows. We start in Sec.~\ref{sec:stabilizer} by reviewing the stabilizer formalism for $\dd$-level systems. In Sec.~\ref{sec:colour_codes} we describe how to generalize a qubit color code in arbitrary spatial dimensions to a qudit color code by employing the notion of star-conjugation. In Sec.~\ref{sec:m_matrices} we review and generalize the triorthogonal matrix construction. In Sec.~\ref{sec:transversal} we derive conditions on the lattice that must hold for a transversal non-Clifford gate to be implementable on the code. Sec.~\ref{sec:gauge_fixing} explains how the Hadamard can be implemented on the same lattice, although not in the color code, using the technique of gauge fixing. In Sec.~\ref{sec:gauge_color_codes} we show gauge color codes can naturally be defined for all the codes we have introduced, inheriting the favorable features of the qubit gauge color codes. We conclude in Sec.~\ref{sec:conclusion}.

\section{Qudit Stabilizer Codes and the Clifford Hierarchy}\label{sec:stabilizer}

We will consider $\dd$-level quantum systems (qudits) as the building blocks for the constructions of  quantum codes~\cite{Gottesman99}. Unless stated otherwise, the qudit dimension $\dd$ is assumed to be any integer greater than two. The conventional basis states for the $\dd$-level system are taken to be $ \ket{j} $, for $ j\in\mathbb{Z}_{\dd} $. The single qubit Pauli matrices $X$ and $Z$ have natural extensions in higher dimensions~\cite{Gottesman99}. The qudit analogues are
\begin{eqnarray}
X = \sum\limits_{j \in \mathbb{Z}_{\dd}} \ket{j+1} \bra{j}, \hspace{5mm} Z = \sum\limits_{j \in \mathbb{Z}_{\dd}} \omega ^j \ket{j} \bra{j},\nonumber\\
\end{eqnarray}
where $\omega = e^{\frac{2\pi i}{{\dd}}}$. With a slight abuse of terminology, we shall say that two operators $A$ and $B$  ``$\omega$-commute'' if $AB=\omega BA$ and note that $\omega$-commutation holds for $X$ and $Z$. These generalized operators simplify to the familiar Pauli operators for $\dd=2$. 

The single qudit Pauli group is generated (up to  global phases) by $ X$ and $Z$ and the $ n $-qudit Pauli group $\mathcal{P}^{\otimes n} $ is the $ n $-fold tensor product of the single qudit Pauli group. Consider an abelian subgroup $\mathcal{S} \subset \mathcal{P}^{\otimes n}$, such that $ \omega^{j} \id \notin\mathcal{S} $ for nonzero $ j $, then we say that $ \mathcal{S} $ is a stabilizer group, and we refer to its elements as the \textit{stabilizers}. The stabilizer group defines an error correcting code with codewords corresponding to states $ \ket{\psi} $ that are stabilized by the stabilizers, i.e. $ S\ket{\psi}=\ket{\psi} $ for all $ S\in\mathcal{S} $.

The logical operators correspond to the set of operators that commute with $\mathcal{S}$ but are not contained in it. A pair of logical operators $\bar{X}_i$ and $\bar{Z}_i$ $\omega$-commute with each other and hence encode one qudit. 

Gottesman and Chuang~\cite{GottChuang99} introduced a classification of quantum gates known as the \textit{Clifford hierarchy} (CH), which can be defined recursively~\cite{Yoshida14} as 
\begin{eqnarray}
\mathcal{P}_{l} = \{ U | P^{\dagger}UPU^{\dagger} \in \mathcal{P}_{l-1} \hspace{2mm} \forall \hspace{2mm} P \in \mathcal{P}_1 \} ,\nonumber\\
\label{eqn:CH}
\end{eqnarray}
where $\mathcal{P}_1$ is the Pauli group. For example, $\mathcal{P}_2$ is the group of operators that leave the Pauli group invariant under conjugation and is called the Clifford group. We shall describe some important Clifford group gates $H$, $S$ and $\Lambda(X)$ below. In prime qudit dimensions, these are known to generate the whole Clifford group~\cite{Gottesman99,Clark06}. The gate $ H $ is the qudit version of the Hadamard gate (also known as the discrete Fourier transform),
\begin{equation}\label{Eq:QuditH}
H = \frac{1}{\sqrt{d}} \sum_{j,k \in \mathbb{Z}_{\dd}} \omega^{jk} \ket{j}\bra{k},
\end{equation}
the $ S $ gate is the generalisation of the qubit $ \pi/4 $-phase gate,
\begin{equation}\label{Eq:QuditS}
S = \sum_{j \in \mathbb{Z}_{\dd}} \omega^{j^2} \ket{j}\bra{j},
\end{equation}
and $ \Lambda(X) $ is the controlled-$ X $ gate (also known as the SUM gate),
\begin{equation}\label{Eq:QuditCX}
\Lambda(X) = \sum_{j,k \in \mathbb{Z}_{\dd}} \ket{j}_c\ket{k\oplus 1}_t\bra{j}_c\bra{k}_t,
\end{equation}
where $c$ and $t$ are the control and target qudits, respectively. 

The set of gates $\mathcal{P}_l$ in the hierarchy contains all gates from lower levels of the hierarchy. %
To refer to gates in level $l$ of the hierarchy but not level $l-1$ we shall say that the level $l$ of the hierarchy is the lowest level of the hierarchy for which the gate is a member. 

In prime dimensions, it is known that it suffices to supplement the Clifford group with just one non-Clifford gate from the third level of the CH in order to obtain a universal set of gates~\cite{Campbell12}. Such a gate is not unique, and in the qubit case, the $T$ gate $\textrm{diag}(1,e^{i \pi/4})$ is usually chosen for this purpose. For the qudit case, we choose the following particularly convenient definition for the $T$ gate, which is valid in all dimensions except when $\dd = 2,3,6$~\cite{Zhu10, Howard12},
\begin{equation}
T = \sum_{j \in \mathbb{Z}_{\dd}} \omega^{j^3} \ket{j}\bra{j}.
\label{Eq:QuditT}
\end{equation}

It is a consequence of lemmas \ref{rlemma1} and \ref{rlemma2} (below) that for $d\neq2,3,6$ this gate is non-Clifford and inhabits the third level of the CH. In $\dd=3$ and $\dd=6$, the gate defined in equation~\eqref{Eq:QuditT} is not non-Clifford, as it reduces to the Pauli $Z$ gate since $j^3=j$ mod $3$ and mod $6$. However, the following definition provides a suitable alternative $T$ gate for these dimensions. The gate is non-Clifford and in the third level of the CH:
\begin{equation}
T_{3,6} = \sum_{j \in \mathbb{Z}_{\dd}} \gamma^{j^3} \ket{j}\bra{j}.
\label{Eq:QuditT36}
\end{equation}
where $\gamma^3=\omega$ and where the function in the exponent $j^3$ is evaluated in regular arithmetic (or equivalently modulo $3d$).
When we refer to ``the $T$ gate'' in this paper we will always mean a gate of the form of equations~\eqref{Eq:QuditT} or~\eqref{Eq:QuditT36}, depending on the qudit dimension under consideration. For notational convenience we suppress the dependence of $T$ on $d$, since it will always be clear by the context which $T$ gate we require.

Notice how the $ T $ gate has a cubic power in the exponent of $\omega$, in contrast to the quadratic power in the case of the $S$ gate. 
For prime dimensions, the first investigation characterising all the phase gates from the third level of the CH was performed by Howard and Vala~\cite{Howard12}. 
In general, there is a close correspondence between the order of the polynomial in the exponent of $\omega$ and the lowest level of the CH the phase gate belongs to.  Let us define the following family of phase gates in terms of a polynomial function $f_r(j)$ of degree $r$, such that $r \leq \dd$ with coefficients $a_{m} \in \mathbb{Z}_{\dd}$, so $f_{r}(j):=\sum_{m=0}^{r}a_m j^m$: 
\begin{equation}
R_{f_{r}} = \sum_{j \in \mathbb{Z}_{\dd}} \omega^{f_{r}(j)} \ket{j}\bra{j},
\label{Eq:QuditR}
\end{equation}
one can then prove the following useful lemmas.

\begin{lm}\label{rlemma1}
For all $\dd$, all $r\leq d$ and all functions $f_{r}(j)$, the gate $R_{f_{r}}$ is in the $r$th level of the Clifford hierarchy.
\end{lm}

\begin{proof}
The proof of this is simple and concise. We begin by calculating 
\begin{equation}\label{hierarchylemma}
    R_{f_{r-1}} = X^{\dagger} R_{f_{r}} X R_{f_{r}}^{\dagger} ,
\end{equation}
where $f_{r-1}$ is a new function 
$    f_{r-1}(j)=f_{r}(j+1)-f_{r}(j) =r a_r j^{r-1} + \cdots $.
Expanding out, the degree $r$ terms cancel, so the leading term is $r a_r j^{r-1}$ and extra terms are all degree $r-2$ or smaller. We now observe that if $r=1$, $R_{f_{r}}$ is a Pauli operator. Using the definition of the Clifford hierarchy in equation~\eqref{eqn:CH}, the lemma follows by induction.
\end{proof}

\begin{lm}\label{rlemma2}
For all $\dd$, all $r\leq d$ and all functions $f_{r}(j)$ satisfying $r! a_r \neq 0 \pmod{d}$, the gate 	$R_{f_{r}}$ is not in the $(r-1)$th level of the Clifford hierarchy.
\end{lm}

\begin{proof}
If the gate is in the $(r-1)$th level of the hierarchy, then applying the inductive transformation in the previous proof $r-1$ times must return a Pauli operator, and thus applying it $r$ times results in an operator proportional to the identity. By chaining together transformations of the form of equation~\eqref{hierarchylemma} $n$ times, we obtain $    f_{1}(j)=r! a_r j + c $. If the corresponding operator is proportional to the identity then this function must be constant, and thus $r! a_r=0$. 
\end{proof}

These two lemmas provide us with a simple way to generate gates at all levels of the Clifford hierarchy as required. They fail when  $r! a_r=0$, which was the case above for $r=3$ and $\dd=2,3,6$. In those exceptional cases, gates can be discovered by moving to higher roots of unity as illustrated by equation~\eqref{Eq:QuditT36}.

\section{Qudit  color codes} \label{sec:colour_codes}

Color codes~\cite{Bombin06} are a class of topological qubit stabilizer codes that be defined on a topological space of any spatial dimension~\cite{Bombin07a} $\DD\ge2$.  The $\DD$-dimensional manifold is celluated into a lattice of objects called $k$-\textit{cells} for all spatial dimensions $0 \le k \le \DD$.  For example, vertices are $0$-cells, edges are $1$-cells and a $2$-cell is a cell defined on the faces, or plaquettes, of the lattice.  We define 
\begin{defin}
\label{mucolex} 
A lattice $\mathcal{L}$ is called a $\DD$-colex (colex is short for ``color complex'') whenever 
\begin{enumerate}
    \item it is a celluation of an orientable $\DD$-dimensional manifold without a boundary; and
    \item every vertex has $(\DD+1)$ neighbors ($(\DD+1)$-valency); and
    \item the $\DD$-cells are $(\DD+1)$-colorable.   
\end{enumerate}
\end{defin}
Any $\DD$-colex defines a qubit color code. For example, the smallest 3-dimensional color code is a 15-qubit code defined on the lattice illustrated in figure~\ref{fig:tetra}(a). The stabilizer group is generated by face operators (with a $Z$ operators assigned to the vertices around the face or 2-cell) and cell operators (with an $X$ operators assigned to the vertices around each 3-cell), see figure~\ref{fig:tetra}(b). Our results show a $\DD$-colex also defines a qudit color code.  An alternative description of a $\DD$-colex is that its dual $\mathcal{L}^*$ is a simplical lattice.  This is lattice where each $k$-cell is a simplex, an object with $k+1$ vertices. In the dual, we have conditions: ($2^*$) every $\DD$-cell has $(\DD+1)$ neighbors; ($3^*$) the vertices are  $(\DD+1)$ colorable.  

Such codes can be constructed~\cite{Bombin13}  by starting with a closed hyperspherical lattice and then removing a vertex to ``puncture'' the surface. Alternatively, as we show in Sec.~\ref{sec:code_distance}, one can construct them directly by defining a suitable boundary on a regular lattice structure.  Such a code encodes a single qubit and has the attractive feature of a transversal non-Clifford $T$ gate. The qudit color codes have many useful properties. For more details, we refer the reader to~\cite{Bombin13,Bombin06,Bombin07a}.

The stabilizer generators are guaranteed to commute in this construction. Those of the same type ($X$ or $Z$) trivially commute, while those of different types commute since they always meet in pairs. As already remarked, $X\otimes X$ and $Z\otimes Z$ commute only in $\dd=2$.  However, by replacing one operator of each pair with its complex conjugate we produce a pair of operators that commute for all $\dd$. These are $X\otimes X^*$ and $Z\otimes Z^*$. To define color codes in all qudit dimensions, we need, therefore, to find a construction which allows us to take advantage of the commutation of $X\otimes X^*$ and $Z\otimes Z^*$. This can be achieved by identifying and exploiting a bipartition or bicoloring of the graph defining the lattice. To avoid confusion with other colorings of the lattice important for color codes we call this the star-bipartition. 
	 
We note that a similar construction was defined (for prime power qudit dimension only) in~\cite{Sarvepalli10}. Here we go further and show that  the star-bipartition is the starting point for an identification of a broad family of transversal gates on these codes, including non-Clifford gates saturating Bravyi and K\"onig's bound.

\subsection{Star-bipartition}

The  star-bipartition is a bipartition of the color code lattice $\mathcal{L}$. It thus divides the set of lattice vertices into starred and unstarred vertices where neighboring vertices always belong to different sets. Here we prove the following useful bipartition lemmas
\begin{lm}[Star-bipartition lemma]
\label{StarLem}
Let $\mathcal{L}$ be a $\DD$-colex, then its vertices can be 2-colored into starred and unstarred sets, $v_{\star}$ and $v_{\bullet}$, respectively.
\end{lm}
\noindent
Recall that the $\DD$-colexes have several properties listed in Def.~\ref{mucolex} that are essential to proving the above lemma. Using $|\ldots|$ to denote the number of elements in a set, we also have
\begin{lm}[Starring of cells lemma]
\label{StarCellLem}
Let $\mathcal{L}$ be a $\DD$-colex, so that its vertices are partitioned into $v_{\bullet}$ and $v_{\star}$ according to lemma~\ref{StarLem}.  It follows that 
\begin{enumerate}
  \item $|v_{\star}|=|v_{\bullet}|$; and
  \item any $k$-cell $C$ with $0<k \leq \DD$ contains a equal number of vertices from each partition, so $| C \cap v_{\star} | = | C \cap v_{\bullet} | $.\end{enumerate}
\end{lm}
Lemmas~\ref{StarLem} and~\ref{StarCellLem}, or variants thereof, were already proved in prior research~\cite{Bombin13,Bombin06,Bombin07a,Kubica14}.  These results are especially important to qudit color codes and play a fundamental role, and so, for completeness, we present our own proofs.  Our presentation of lemma~\ref{StarLem} has the merit of being more concise than that of Ref.~\cite{Kubica14}. 

Before presenting our proof, we review the concept of orientation from topology theory~\cite{Lee}.  An orientation on a $\mu$-simplex is an ordered list of its vertices, with orientations considered equivalent whenever they differ by a permutation generated from an even number of transpositions (a transposition is swap of two elements). All permutations can be generated by either an even or odd number of transpositions, and so there are two possible orientations for a $\mu$-simplex. The $\mu$-simplex contains subsimplices that can be obtained by removing a single vertex. This vertex removal provides an induced orientation on the subsimplex, such that removing $v_j$ from a simplex with orientation $\{ v_1, v_2, \ldots, v_{j-1}, v_j, v_{j+1}, \ldots, v_\mu , v_{\mu+1} \}$ induces an orientation $\{ v_1, v_2, \ldots, v_{j-1}, v_{j+1}, \ldots, v_\mu\,, v_{\mu+1} \}$.  Consider two oriented $\mu$-simplices with a common subsimplex (that is, they are are neighbors).  We say their orientations are consistent if they induce opposite orientations to their common subsimplex. An entire lattice of simplices is said to be orientable if there exists a choice of orientations such that all pairs of simplices are consistently orientated.  An example of a consistently orientated lattice is shown in Fig.~\ref{fig:lemma3}. Orientablity of a lattice is a topological feature that depends on the underlying manifold and from this definition one can show that many familiar manifolds are orientable, including Euclidean, spherical and toroidal manifolds.  We now employ these concepts in our proof.

\begin{figure}
\includegraphics[width=1\columnwidth]{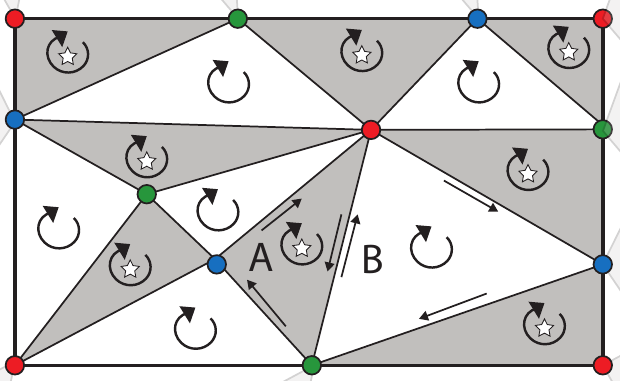} \centering
\caption{(Color online) A patch of an orientated simplicial lattice with 3 colored vertices. The orientation is visually represented by a clockwise symbol.  Formally, the orientation is an ordered list of vertices $\{x,y,z\}$ moving clockwise round the simplex.  For neighboring simplices $A$ and $B$, we also show the induced orientation on their edges.  For an edge orientation $\{x,y\}$, we illustrate this with an arrow going from $x$ to $y$. The simplices $A$ and $B$ induce opposite orientations to their common edge, and so are consistently orientated. If a simplex has an orientation $\{a,b,c \}$ with colors $\{ \mathrm{red},  \mathrm{green},  \mathrm{blue} \}$ then we label it starred, and otherwise unstarred.  The star-bipartition lemma further formalises this notion of starring in terms of color and orientation.}
\label{fig:lemma3}
\end{figure} 
 
\begin{proof}
To prove the star-bipartition lemma, we switch to the dual lattice $\mathcal{L}^*$.  In the dual picture, we need to show that the $\DD$-dimensional cells can be 2-colored.  Property 2 of Def.~\ref{mucolex} ensures that the dual lattice is simplical. Property 3 of Def.~\ref{mucolex} ensures the dual vertices are $(\DD+1)$ colored.  This coloring can be described by a map $\mathcal{C}: \mathcal{L}^* \rightarrow \mathbb{Z}_{\mu+1}$ that assigns color label $\mathcal{C}(v_j)$ to vertex $v_j$.  Given a simplex $\sigma$ with an orientation, say $\{ v_1, v_2,\ldots v_{\mu+1} \}$, we can apply the color map to the orientation to obtain a colored-orientation $\{ \mathcal{C}(v_1), \mathcal{C}(v_2),\ldots \mathcal{C}(v_{\mu+1}) \}$.  This defines a bipartition (the star-bipartition) with two simplicies belonging to the same partition if their colored-orientation differs by an even permutation. The remainder of the proof establishes that neighbouring simplicies belong to different partitions.  Two neighbouring oriented simplices $A$ and $B$ always share a common $(\mu-1)$ subsimplex.  Let us call the vertices they share in common $\{ v_1, \ldots, v_\mu \}$.  Simplex $A$ also contains a vertex $v_A$ and we label its orientation $\{ v_1, \ldots, v_\mu, v_A \}$. Similarly, simplex $B$ contains additional vertex $v_B$ and has orientation $\{ v_{\Pi(1)}, v_{\Pi(2)},\ldots v_{\Pi(\mu)}, v_B  \}$ where $\Pi$ is some permutation of vertex labels.  No generality is lost by placing $v_A$ and $v_B$ last in their respective orientations as this can always be achieved with an even number of transpositions (assuming $\mu>1$). Since the manifold is orientable, we can assume a consistent orientation on their common subsimplex.  We see that $A$ induces the orientation $\{ v_1, \ldots, v_\mu\}$ and $B$ induces the orientation $\{ v_{\Pi(1)}, \ldots, v_{\Pi(\mu)}\}$, and so consistency demands that $\Pi$ is an odd permutation of the vertex labels.  Next, we show this lifts to odd permutation in the color-orientation of $A$ and $B$.    Simplex $A$ has color-orientation $\{ \mathcal{C}(v_1), \ldots, \mathcal{C}(v_\mu), \mathcal{C}(v_A) \}$ and simplex $B$ has color-orientation $\{ \mathcal{C}(v_{\Pi(1)}), \ldots, \mathcal{C}(v_{\Pi(\mu)}), \mathcal{C}(v_B) \}$.  Since the lattice is $(\mu+1)$ colored, we know $v_A$ and $v_B$ must possess the same color, which is whatever color is absent from their common subsimplex.  Permutation of vertex labels results in a corresponding permutation of colors. Therefore, the color-orientations of $A$ and $B$ differ by an odd permutation of their first $\mu$ elements, which corresponds to an overall odd permutation. Since all neighboring $\DD$-simplicies must have opposite color-orientations, they belong to separate partitions of our star bipartition, proving the lemma. 
\end{proof}

Let us now turn to lemma~\ref{StarCellLem}.

\begin{proof}
We shall begin with part \textit{1} of the lemma. Consider a bipartite lattice that is $r$-valent. We count the number of edges $N_E$.  Every edge is incident to one, and only one, starred vertex. Furthermore, each vertex is contained in $r$ edges, and so the total count is $N_E = |v_{\star}| r$.  The same argument applies to unstarred vertex, so $N_E = |v_{\bullet}|r$.  Equating $|v_{\star}| r = |v_{\bullet}| r$, we see that $r>0$ entails $|v_{\star}|  = |v_{\bullet}| $.  Since a $\DD$-colex is $(\DD+1)$-valent, we have proven part \textit{1} of lemma~\ref{StarCellLem}.  

Consider any $k$-cell $C$. It defines a sublattice $\mathcal{L}_C$ with same vertex set as $C$.  If $\mathcal{L}_C$ is a $k$-valent with $k>0$, then the above argument applies and $C$ contains equal number of starred and unstarred vertices.  The following reasoning parallels that in Ref.~\cite{Bombin13}.  We show explicitly that $\mathcal{L}_C$ is a $k$-colex for the  $k=\DD-1$ case, but any $k$ is reached by iteratively applying the argument down to the desired $k$. Again, we must switch to the dual lattice where $\mathcal{L}_C^*$ is obtained by removing a single vertex $C^*$ from $\mathcal{L}^*$.  The sublattice $\mathcal{L}_C^*$ retains all simplices containing $C^*$, but with their dimension reduced by the removal of $C^*$.  Therefore $\mathcal{L}_C^*$ is a simplical lattice of dimension $(\DD-1)$ and so $\mathcal{L}$ is $\DD$-valent.  Our proof only requires the correct valency, which we have shown, but note that $\mathcal{L}_C$ also has the correct coloring for a $(\DD-1)$-colex.   Specifically, if $C'$ and $C$ intersect on some $(\DD-1)$-cell of $\mathcal{L}_C$, then we assign it the color of $C'$. 
\end{proof}

The above results concern a lattice without a boundary.  However, if one punctures the code removing a starred vertex then we have
\begin{cor}
\label{cor_punc}
    Let $\mathcal{L}$ be puncturing of a $\DD$-colex, with the inherited bipartition into $v_{\star}$ and $v_{\bullet}$.  Then
    \begin{enumerate}
    \item $|v_{\star}|=|v_{\bullet}|-1$; and 
  \item any $k$-cell $C$ with $0<k \leq \DD$ contains a equal number of vertices from each partition, so $| C \cap v_{\star} | = | C \cap v_{\bullet} | $.
    \end{enumerate}
\end{cor}
Property \textit{1} is immediately follows from point \textit{1} of lemma~\ref{StarCellLem} as a single starred vertex has been removed.  Property \textit{2} is identical to its partner in lemma~\ref{StarCellLem}. The punctured lattice only keeps cells that did not contain the punctured vertex, and so the property is directly inherited.  The dual lattice $\mathcal{L}^*$ has been an essential proof tool for these lemmas, but in the rest of the paper we shall only consider the primal lattice.


\subsection{Qudit color codes in arbitrary spatial dimension}

\begin{figure}
\includegraphics[width=1\columnwidth]{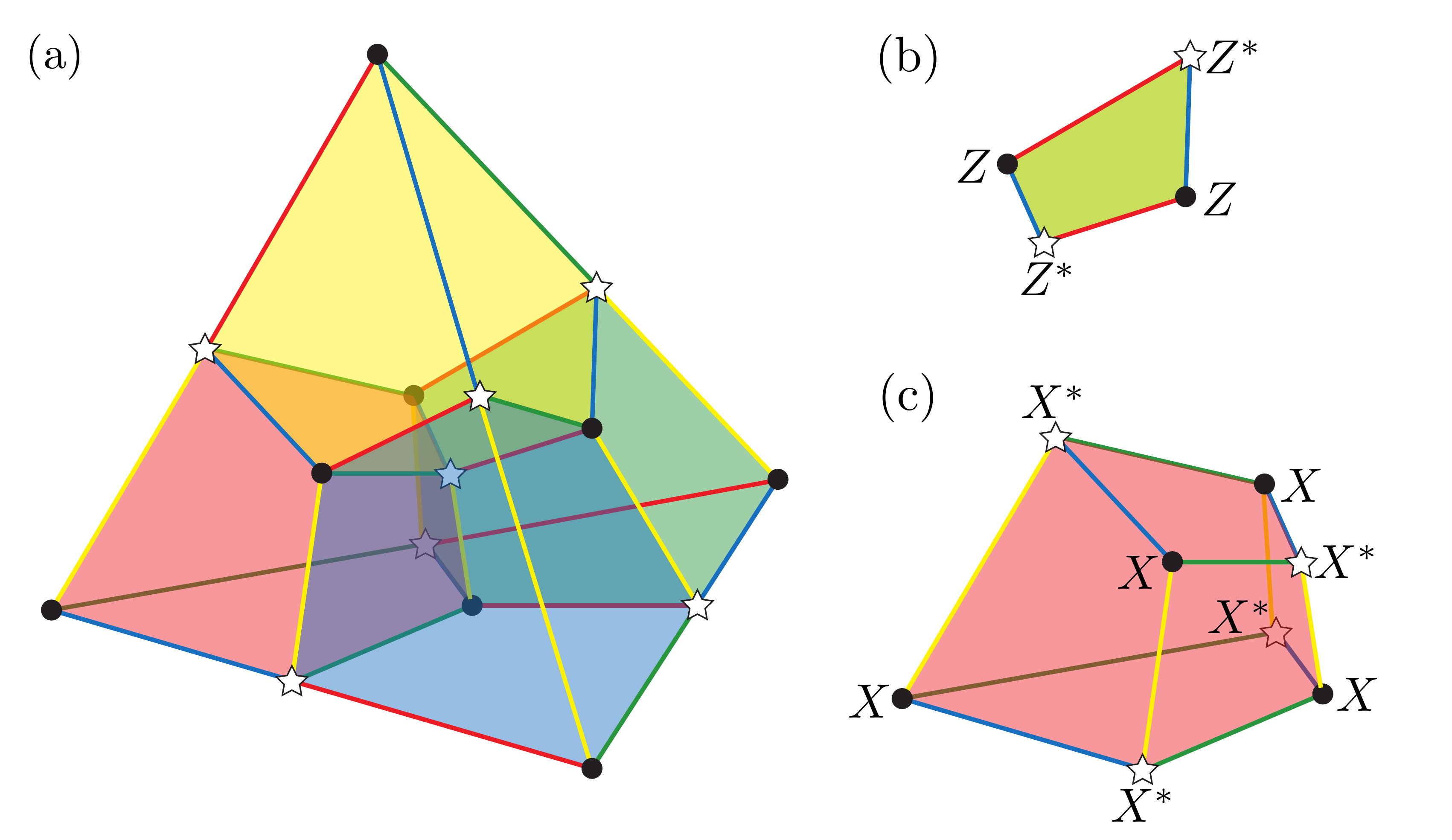} \centering
\caption{(Color online) (a) The smallest instance of the qudit 3D color code with the $X$-type stabilizers colored red, green, blue and yellow. The 1-cells (edges) of the code are also colored. The vertices are also colored so the set $v_{\bullet}$ is represented by black circles and the set $v_{\star}$ is represented by white stars. (b) A single $Z$ stabilizer of the tetrahedral 3D color code. The plaquette can take the color yellow if considered as a face of the green 3-cell, and vice versa. (c) A single $X$ stabilizer of the tetrahedral 3D color code. }
\label{fig:tetra}
\end{figure}


Before we write down the stabilizer generators for these codes, we shall use the star-bipartition to define notation for an important family of transversal operators which we call star-conjugate transversal.
Throughout this paper, when discussing color codes we identify qudits with vertices in the lattice defining the code. Let $v$ denote a set of vertices (and thus the corresponding qudits). Let us introduce the notation

\begin{equation}
U[v]=U ^{\otimes|v|}\end{equation}
to denote the tensor product of unitary $U$ acting on each qudit identified with the vertices in $v$.

We call an operator \textit{star-conjugate transversal} when it has the following form
\begin{equation}
\tilde{U} = U[v_{\bullet}] \otimes U^{*}[v_{\star}]. \label{eqn:U_trans} 
\end{equation}
In other words,  $\tilde{U}$ consists of $U$ applied to all unstarred vertices and $U^*$ applied to all starred vertices.

To define qudit color codes we need to introduce two types of stabilizer generators which are defined with respect to cells of different dimensions within the lattice. As in the qubit case, we associate $Z$-type stabilizer generators with the $\DD^{\prime}$-cells of the lattice and  $X$-type stabilizers with the $(\DD-\DD^{\prime}+2)$-cells~\cite{Bombin07a}. In this section, for simplicity of presentation, we will take the example of $\DD^{\prime}=2$ so the $Z$-type stabilizers are associated with plaquettes ($2$-cells) and $X$-type stabilizers are associated with $\DD$-cells, although the construction holds in the more general case, which we shall consider later. For example in 3D, the $X$ stabilizers act on the vertices contained in a 3-cell of the lattice, and the $Z$ stabilizers act on the vertices contained in a 2-cell (plaquette). This can be seen in figure~\ref{fig:tetra}.
 
Setting $\DD^{\prime}=2$ the stabilizer generators therefore take the form 
\begin{eqnarray}
S_{X,C} &=& X[v_{\bullet} \cap C] \otimes X^{*}[v_{\star} \cap C], \label{eqn:X_stab}\\
S_{Z,P} &=& Z[v_{\bullet} \cap P] \otimes Z^{*}[v_{\star} \cap P]. \label{eqn:Z_stab}
\end{eqnarray}
for all $2$-cells $P$, and $\DD$-cells $C$ of the lattice. Recall that we write the conjugate operator $X^*$ explicitly to emphasis the ubiquity of starring, even though $X$ is a real operator in the computational basis and $X=X^*$.  It is not just the $S_{Z,P}$ operators where starring is non-trivial, but many logical operators also require starring.

Let us denote the group generated by $S_{X,C}$ operators $S_X$ and $S_{Z,P}$ operators by $S_Z$.  All elements of $S_{X}$ commute as do all elements of $S_{Z}$. It remains to show that all elements of $S_{X}$ commute with all elements of $S_{Z}$. This follows from the fact that $X \otimes X^*$ commutes with $Z \otimes Z^*$. The above construction ensures that whenever cell $C$ and cell $P$ overlap they overlap on an equal number of starred and unstarred vertices, which is a point discussed further in Sec.~\ref{moitySec}.  Hence these stabilizer generators commute as required. Note that the qubit $\dd=2$ case is included in our definition.


We noted above that, in the constructions we consider, the code will contain one more unstarred qudit than there are starred qudits. We can thus define the logical encoded Pauli operators for the code star-conjugate transversally.
\begin{eqnarray}
\begin{array}{l}
\bar{X} = \tilde{X} = X[v_{\bullet}] \otimes X^{*}[v_{\star}], \\
\bar{Z} = \tilde{Z} = Z[v_{\bullet}] \otimes Z^{*}[v_{\star}].
\end{array} \label{eqn:log_ops}
\end{eqnarray}
 The fact that $|v_{\star}|=|v_{\bullet}|-1$ ensures that the logical operators satisfy the same commutation properties as $X$ and $Z$. One can verify that these operators commute with the stabilizer operators defined above by recalling that each 2-cell $P$ and and $\DD$-cell $C$ contains an equal number of starred and unstarred vertices.

\subsection{Constructing codes of any code distance} \label{sec:code_distance}

In topological stabilizer codes we expect to be able to increase the number of physical systems encoding the quantum information in order to protect the information more effectively---a property characterized by the code distance. We have remarked already that suitable codes can be constructed by puncturing a hypersphere.  In this section we provide, however, a constructive alternative method to show that in $\DD$ spatial dimensions the color code distance can be made arbitrarily large.

The code distance is the weight of the smallest logical operator. It was convenient above, in equations~\eqref{eqn:log_ops}, to define the logical operator basis star-conjugate transversally. However, multiplying by a subset of the stabilizer group we can localize $\bar{Z}$ to any edge of the polytope formed by the lattice, where such an object is a 1-dimensional sub-manifold of the lattice between two of the vertices of the polytope. Each such vertex is contained in only a single $\DD$-cell  (as opposed to those contained in two or more cells of the lattice), and hence there are exactly ($\DD+1$) vertices of the polytope. 

In $\DD$ spatial dimensions, an edge of the polytope is comprised of 1-cells of two different colors. Similarly, 2-cells contain 1-cells of two colors. Beginning from the transversal definition of the operators in  equations~\eqref{eqn:log_ops} and multiplying $\bar{Z}$ by the subset of $Z$ stabilizers defined on 2-cells containing 1-cells of the same pair of colors, the logical operator is localised to the edge of the polytope uniquely defined as containing 1-cells matching this pair of colors. The constraints placed on the subset of stabilizers chosen, and the initial choice of transversal operator basis ensures that the minimum length of a string-like logical operator defined in this way is the length of the edge of the polytope colored as indicated. The minimum length of a string-like $\bar{Z}$ in such a definition, and hence the code distance, is then length of the shortest edge of the polytope. Thus by showing that the length of an edge of the polytope can be made arbitrarily large we demonstrate the ability to create codes of any distance. 

We begin with an example in 3D and then generalize the argument to higher spatial dimensions. The particular choice we make for this example is a 4-colorable tiling of \textit{truncated octahedra}, part of which is illustrated in figure~\ref{fig:truncated_octa}. This is not a unique choice---for an alternative see~\cite{Bombin07a} where the 3D color code is defined on a lattice constructed from cubes and truncated octahedra.

In addition to the two lattices already mentioned, we have identified two further possible choices of regular tilings in 3D Euclidean space on which a color code may be defined. One is the cantitruncated cubic honeycomb, comprising truncated cuboctahedra (polyhedra with square, hexagonal and octahedral faces), truncated octahedra, and cubes. The second is the omnitruncated cubic honeycomb comprising truncated cuboctahedra and octahedral prisms. This list is not exhaustive, for instance many more tilings may exist in curved space. 

\begin{figure}
\includegraphics[width=0.35\textwidth]{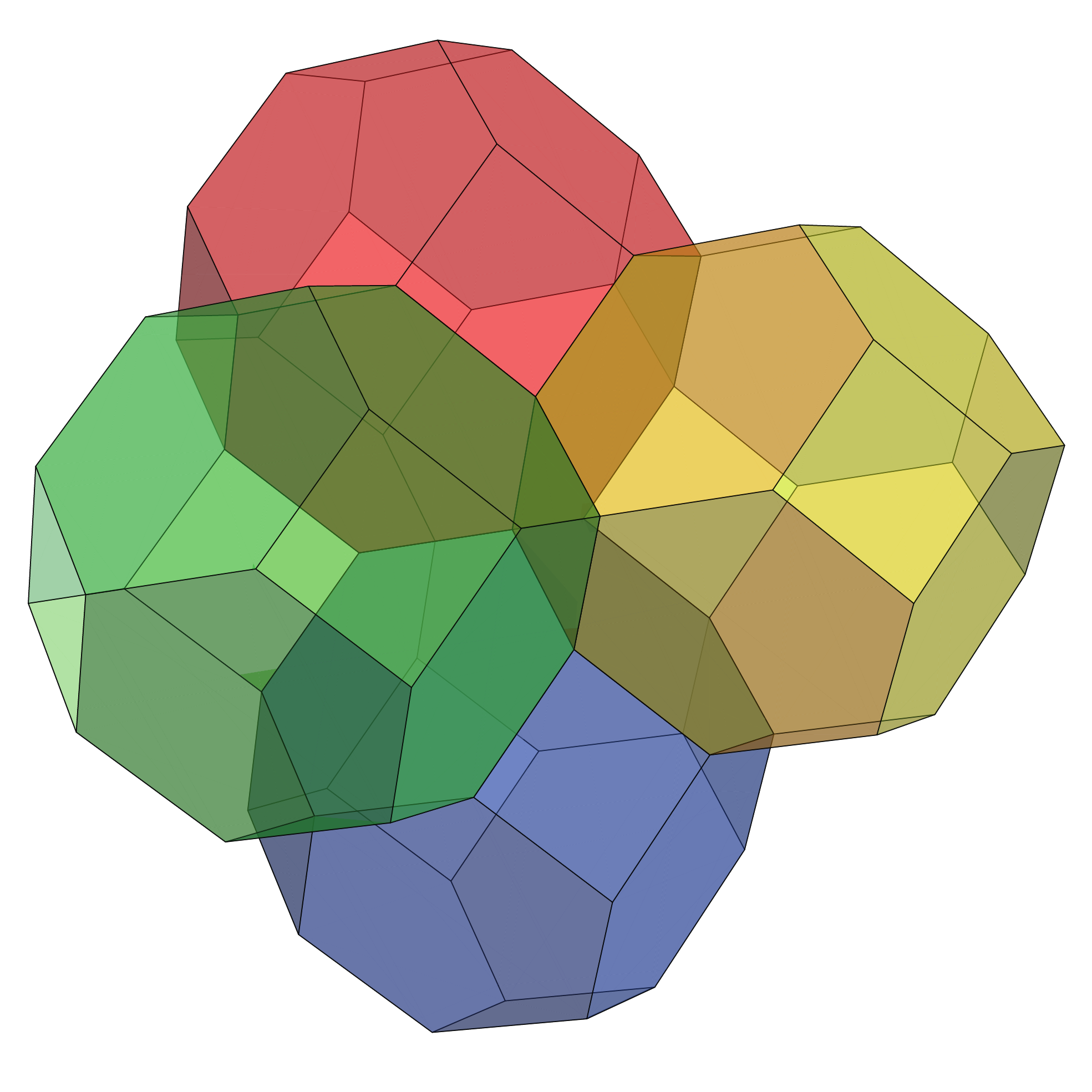} \centering
\caption{(Color online) To increase the distance of the code, we tile 3D space with truncated octahedron cells (with appropriate boundary conditions). Here we show a portion of the lattice. A red cell fits into the hollow formed by the green, blue and yellow cells towards the bottom of the figure, and thus the tiling proceeds, ensuring the 4-colorability of the lattice is preserved.}
\label{fig:truncated_octa}
\end{figure}

The method proceeds by carefully cutting a block of the lattice out, in order to form a tetrahedron, similar to that illustrated in figure~\ref{fig:tetra}(a). The main difference is that the new tetrahedron contains  more qudits and can have an arbitrarily large (though always odd, in order to ensure the $\omega$-commutation of the logical operators) number of qudits contained in its edges. The shape cut from the lattice may not look tetrahedral at first but may be deformed to a regular tetrahedron, with vertices belonging only to a single cell forming the four vertices of the tetrahedron.

In $\DD$ spatial dimensions a suitable tiling of the space as outlined at the beginning of section~\ref{sec:colour_codes} must be found. Besides these requirements placed on the lattice construction, the ``block'' (polytope) cut out of the lattice to form the higher-dimensional analogue of the tetrahedron must adhere to certain rules. 
\begin{enumerate}
\item The polytope formed must have $(\DD +1)$ boundaries, each of a different color. The color of the boundary corresponds to the color of the string-like operators that can end on that boundary.
\item Every edge of the polytope must contain an odd number of qudits. As stated above, this is to ensure the $\omega$-commutation of the logical operators is preserved.
\item There should be exactly $\DD$ vertices that belong only to a single cell of the lattice. These cells should be different colors and these vertices will form the vertices of the polytope. 
\end{enumerate}

\subsection{Error detection}

Errors are detected by measuring all the stabilizer generators. We shall not present a detailed analysis here, since errors and syndrome are related for these codes in a  similar way to the qubit case. For example, in 3 spatial dimensions a single $Z$-type error on an unstarred qudit will result in a measured eigenvalue of $\omega$ in the $X$ stabilizers corresponding to the four cells containing that vertex (see figure \ref{fig:tetra}). Similarly, a single $X$-type error on an unstarred qudit will result in a measured eigenvalue of $\omega$ in the six faces that contain the vertex (see figure \ref{fig:tetra}). In general, a $Z^k$ or $X^k$ error will lead to a measurement outcome  of $\omega^{\pm k}$. Syndromes arising from sets of multiple errors can be calculated using standard techniques.

A classical algorithm called a \textit{decoder} may be used to interpret the syndrome and infer the correction operator that must be applied to return the code to its original state. There are many examples of decoding algorithms for 2D qubit color codes that have been developed recently~\cite{Wang09,Katzgraber09,Ohzeki09,Andrist11,Delfosse14,Landahl11,Sarvepalli12}, but very little is known for the case of higher qudit and spatial dimensions. 

However, in the case of surface codes, it has been shown that some proposed renormalization group (RG) algorithms can be generalized to qudit codes in a straightforward manner~\cite{Anwar14, Poulin14}. Moreover, recent work for the 3D fault-tolerant implementation of the qudit surface code using a hard-decision RG decoder~\cite{Watson14} suggests that adaptations of such algorithms for higher spatial dimensions should be possible, however, the error correction thresholds would likely be degraded as the spatial dimension increases due to the larger stabilizer generators. 

We shall not present a decoder in this paper,  but remark that, unlike in qubit codes where neighboring errors lead to cancellation of syndromes in between, such cancellation becomes unlikely with increasing qudit dimension, and thus in general more information is available to the decoder. In fact, it has been shown that it is possible to design decoders to exploit this additional information to obtain higher thresholds with respect to qubit codes~\cite{Watson14}. 

For the case of gauge color codes, very recently a hard-decision RG decoder \cite{Brown15} was implemented for the qubit 3D construction. As with the color code, we expect a generalization of such an algorithm to gauge color codes in arbitrary spatial and qudit dimensions to be straightforward.

\section{\mo{} matrices and codes}\label{sec:m_matrices}

Brayvi and Haah introduced a powerful framework for defining codes with a transversal non-Clifford gate;  codes defined in terms of triorthogonal matrices. Here we introduce a significant generalisation to Bravyi and Haah's approach, extending to qudits codes with transversal gates from higher up the Clifford hierarchy. 

\subsection{Matrix representation of quantum codes} \label{sec:BH_triortho}

We begin by reviewing how Bravyi and Haah represent CSS codes with matrices~\cite{BravHaah12}. They use a matrix $G$ over $\mathbb{Z}_2$ that has linearly independent rows under modulo 2 arithmetic. The matrix $G$ is broken up into two blocks $G_1$ (with $k$ rows) and $G_0$ (with $s$ rows).  This defines a quantum code with $k$ logical qubits.  The elementary logical basis state $ \ket{0_{L}} $ is written
\begin{align}
    \ket{0_{L}} &= \frac{1}{\sqrt{2^s}} \sum_{f \in  \mathrm{span}[G_{0}]  }  \ket{f}\nonumber\\ 
    &= \frac{1}{\sqrt{2^s}} \sum_{y \in \mathbb{Z}_{2}^{s}  } \ket{ y^{T} \cdot G_0 }.  \label{eqn:log_0_state}
\end{align}
where ``$\cdot$'' denotes matrix multiplication modulo $\dd$, and we take advantage of the fact that the rows of $G_{0}$ are linearly independent to represent the terms in the superposition by a row vector $y$ transposed and multiplied with $G_0$. For the other logical computational basis states  $\ket{ x_{L} }$ where $x \in \mathbb{Z}_{2}^{k}$, we have
\begin{eqnarray}
    \ket{x_{L}} &=& \frac{1}{\sqrt{2^s}} \sum_{y \in \mathbb{Z}_{2}^{s}   } \ket{ y^{T}\cdot  G_{0} \oplus x^{T} \cdot G_{1}},  \nonumber\\
    \label{eqn:log_state}
\end{eqnarray}
where addition of row vectors is elementwise modulo 2 as acknowledged by the symbol $\oplus$.  Note that in the special case that $k=1$, $x$ is a scalar (or a $1\times 1$ row vector) and the transpose operation is trivial.  For qudits, quantum codes are defined by a matrix $G$ taking elements from $\mathbb{Z}_d$, so that
\begin{eqnarray}
    \ket{x_{L}} &=& \frac{1}{\sqrt{ d^s }} \sum_{y \in \mathbb{Z}_{d}^{s}   } \ket{ y^{T}\cdot  G_{0} \oplus x^{T} \cdot G_{1}}.  \nonumber\\
    \label{eqn:log_state2}
\end{eqnarray}
The key differences are that we now sum over all $y \in \mathbb{Z}_{d}^{s}$ and that arithmetic is modulo $d$.

In the CSS formalism, a quantum code has a stabilizer group which is generated by a product of two distinct generating sets, the $X$-stabilizer generators, which are tensor products of $X$ and $I$ alone and the $Z$-stabilizer generators, which are tensor products of $Z$ and $I$ alone. In addition, one must define logical operators, and in the CSS formalism, logical encoded ${X}$ operators consist of a tensor product of $X$ and (optionally) $I$ alone, and the logical encoded ${Z}$ consists of a tensor product of $Z$ and (optionally) $I$ alone. Once the $X$-stabilizer generators and logical $\bar{X}$ operator (or operators if there are multiple encoded qubits) are defined, the remaining $Z$-stabilizer generators are fixed by conjugation relations which can be elegantly captured via the use of dual codes from classical coding theory.  In the qudit setting, $G_0$ will again define the $X$-stabilizer generators for the code, via the mapping $k$ to $X^k$, e.g. the row vector $000123$ defines the operator $I\otimes I\otimes I\otimes X\otimes X^2 \otimes X^3$, and $G_1$ will similarly define the logical $X$ operators for the code.

\subsection{Defining \moity{}}

Bravyi and Haah~\cite{BravHaah12} prove that a so-called triorthogonal code supports a transversal non-Clifford gate in the 3rd level of the Clifford hierarchy.  We now generalize Bravyi and Haah's construction. First, we focus on generalising from qubits to qudits. We defined qudit color codes using a star-bipartition, and it is useful to incorporate that into the definition.  We introduce the star sign-flip matrix $F$. Each column of matrix $G$ corresponds to a qudit, and let us order these columns such that the first $p$ columns correspond to unstarred qudits, and the latter $n-p $ columns to starred qudits. We then define the  $n\times n$ matrix $F$:
\begin{equation}
F=\textrm{diag}(\underbrace{1,\ldots,1}_{\textrm{$|v_\bullet|$ entries}},\underbrace{-1,\ldots,-1}_{\textrm{$|v_\star|$ entries}})	,
\label{eqn:Fdef}
\end{equation}
where the first $|v_\bullet|$ elements on the diagonal are $1$ and the remaining elements on the diagonal are $-1$. We note that the ordering of the columns here is arbitrary, so we do not lose any generality by ordering the columns in this manner.
The purpose of $F$ will be to flip the sign of the entries of row vectors corresponding to the starred qudits, e.g. if $g=\{ [g]_1, [g]_2, \ldots , [g]_n \}$ is a  row vector, we define $g\cdot F$ as the row vector where $[gF]_j=[g]_j$ for the first $p$ elements (corresponding to unstarred qudits) and  $[gF]_j=-g_j$ for the remaining elements (corresponding to starred vertices).  We also define the weight $|\cdots|$ of a row vector as $| g | = \sum_j[ g ]_j$. 

We now will define a significant generalisation of triorthogonal matrices, which we call \moity{} (pronounced ``$\CD$ star orthogonality").
\begin{defin}
\label{mDefuniv}
An $n\times n$ matrix  $G$ over $\mathbb{Z}_{\dd}$ with a a bipartition of columns into $\{v_{\bullet}, v_{\star} \}$ is \mo{} if both of the following conditions hold.
    \begin{enumerate}
       \item The weight of every elementwise product of any $\CD$ rows of $G F$ (including repeated rows) is equal to  $0 $, except for the following case:
        \item For each row of submatrix $G_1 F$, the weight of the row vector raised to the $\CD$th power elementwise is $1 $.
    \end{enumerate}
where $F$ is defined as in equation~\eqref{eqn:Fdef} and where the matrices' columns are ordered to respect to the star-bipartition $\{v_{\bullet}, v_{\star} \}$.\end{defin}
Notice the above makes no mention of modular arithmetic. If we had instead defined the weights modulo $d$, we would have a weaker notion of \moity{}.  Color codes turn out to satisfy this stronger notion, and so this definition suffices for this paper.  We note that the results below rely on the stronger form to deal with the exceptional cases, e.g. $d=3,6$ for the $T$ gate (see App.~\ref{sec:m-ortho}).  We also remark on the weaker notion to clarify that for $d=2$ it exactly corresponds to the Brayvi-Haah definition of triorthogonality.  Other applications of weak \moity{} include quantum Reed-Muller codes~\cite{Campbell12,Campbell14}.


A more symbolic statement of \moity{} will prove useful in subsequent sections. We make use of of $u \circ v$ symbol to denote element-wise products of vectors $u$ and $v$, such that it has elements
\begin{equation}
	 [g_a \circ g_b ]_{j}  =  [g_{a}]_j [g_{b}]_j ,
	 \label{eqn:circle_prod}
\end{equation}
which generalizes for an arbitrary number of vectors, e.g. for three vectors
\begin{equation}
	 [g_a \circ g_b \circ g_c ]_{j}  =  [g_{a}]_j [g_{b}]_j [g_{c}]_j  .
\end{equation}
Let us now consider a reformulation of points \textit{1} and \textit{2} of Def.~\ref{mDefuniv}.  For any set of $\CD$ rows, we have a list of $\CD$ row indices $\{ a,b,\ldots, y \}$.  \moity{} demands that for all the 
\begin{equation}
  \sum_j   [g_a F \circ g_b F \circ \ldots \circ g_y F ]_{j}   = 0 
\end{equation}
unless all rows are identical and come from $G_1$, in which case the weight is unity.

\subsection{The \moity{} of qudit color codes}
\label{moitySec}

The concept of \moity{} is a powerful tool for studying the properties of qudit color codes due to the following lemma

\begin{lm}\label{Colorcode_morth_lemma}
A qudit color code for any $\dd\ge 2$ defined on a lattice in spatial dimension $\DD$ where $X$-stabilizer generators are defined by $\DD^{\prime}$-cells with $\DD^{\prime}\le\DD$, and
with a star-bipartition $\{v_{\bullet}, v_{\star} \}$  defined by the bipartition of the lattice, is a \mo{} code for all $\CD\le\DD^{\prime}$.
\end{lm}
To prove this, let us first prove a convenient lemma.
\begin{lm}\label{Colorcode_morth_sublemma}
Any \mo{} matrix $G$ which includes, as one of its rows, the all-ones vector is also $\CD^{\prime \star}$-orthogonal for all $\CD^{\prime}<\CD$.
\end{lm}
\begin{proof}
This follows since an elementwise product of $\CD$ vectors including the all-ones vector is equal to the elementwise product of $\CD-1$ vectors excluding the all-ones vector. 
\end{proof}

Since the logical $X$ for all color codes is the transversal $X$ acting on all qudits, the matrix $G_{1}$ for these codes contains the all-ones row. Hence to prove lemma~\ref{Colorcode_morth_lemma} we now only need to prove that a color code in $\CD$ spatial dimensions is \mo{}.

\begin{proof}
Recall that the $X$-generators of the code are defined by $\DD^{\prime}$-cells. Taking as an example the 3D lattice in figure~\ref{fig:tetra}, the $X$-stabilizer generators are defined by 3-cells. Consider the following geometric properties of these cells. When $q$ distinct $\DD^{\prime}$-cells intersect non-trivially, where they meet defines a cell of smaller dimension.  For a general lattice, there is no further restriction on the dimension of this cell.   However, for a $\DD$-colex it is well known~\cite{Bombin13,Bombin06,Bombin07a,Kubica14} that the intersection of $q$ objects of dimension $\DD^{\prime}$ yields either an empty set or a cell of dimension $\DD^{\prime} - q +1$.  For instance, the intersect of two neighboring $\DD^{\prime}$-cells defines a $(\DD^{\prime}-1)$-cell, where three such $\DD^{\prime}$-cells meet defines a $(\DD^{\prime}-2)$-cell, and so on. Where $\DD^{\prime}$ $\DD^{\prime}$-cells meet defines a 1-cell, or lattice edge. For example, in 3D, two adjacent 3-cells meet at a face, and three adjacent 3-cells meet at an edge (four adjacent 3-cells meet at a point). 

We use this geometric fact to prove the theorem. Any product of $\CD$ row vectors in $G_0$  has a geometric representation as the intersection of the vertices of these cells.  This corresponds to either an empty cell or cell of dimension no less than $\DD-\CD+1$.  Each cell in the lattice, of any dimension greater than zero, has an equal number of starred and unstarred vertices (see Cor.~\ref{cor_punc}). When we multiply the product row vector by $F$ we invert the sign of the columns corresponding to the starred vertices, but this corresponds to precisely half of the non-zero elements in the vector. Hence the weight of the resultant vector is zero provided $\DD-m+1 > 0$.  Therefore, the rows of $G_0$ satisfy the conditions for \moity{} whenever $\CD \leq \DD$.

The other case to consider is the product where all $\CD$ vectors in the product are the all-ones vector. The element-wise product results trivially in the all-ones vector. The number of unstarred qudits is one more than the starred qudits (see Cor.~\ref{cor_punc}), hence the weight of this vector after multiplying by $F$ is 1. Together with lemma~\ref{Colorcode_morth_sublemma} this proves lemma~\ref{Colorcode_morth_lemma}.
\end{proof}

Note that the bounds in this lemma are tight. A code whose $X$-stabilizer generators are defined by $\DD^{\prime}$-cells is not $(\DD^{\prime}+1)^{\star}$-orthogonal because $(\DD+1)$ $\DD^{\prime}$-cells meet at a point, and a single vertex represents a vector which does not have zero weight.

\section{Transversal operators} \label{sec:transversal}

Our attention now turns to transversal gates. The qubit color codes are the only family of topological codes known to support transversal gates satisfying the Bravyi-K\"onig bound. We prove in this paper that (apart from a minority of special cases which need to be treated individually) the qudit codes we have defined also saturate this bound. 

\begin{thm}
\label{thm:mainresult}
The qudit color codes on a $\DD$-colex, where $X$-stabilizer generators are defined by $\DD^{\prime}\le \DD$ cells and where  $\DD^{\prime}!\neq 0 \pmod{d}$, support transversal gates in the $\DD^{\prime}$th level of the Clifford Hierarchy (CH) (which are  not in the $(\DD-1)$th level of the CH), saturating the Bravyi-K\"onig bound.
\end{thm}

The qualification $\DD^{\prime}!\neq 0 \pmod{d}$ is due to lemma~\ref{rlemma2}, namely that when $\DD^{\prime}!= 0 \pmod{d}$, gates of the form of equation~\eqref{Eq:QuditR} are in the  $(\DD^{\prime}-1)$th level of the CH (and therefore does not saturate Bravyi-K\"onig) as discussed in Sec.~\ref{sec:stabilizer}. 

For the most important $\DD=3$ case, however, we additionally prove transversality for the exceptional $\dd=3$ and $\dd=6$ cases, and can thus state the unqualified result:

\begin{thm}
\label{thm:mainresultTgates}
The qudit color codes on a $\DD$-colex, where $X$-stabilizer generators are defined by $\DD^{\prime}\le \DD$ cells support transversal non-Clifford gates in the $3$rd level of the Clifford Hierarchy (CH) for all $\DD^{\prime}\ge 3$.
\end{thm}

We believe that the qualification in theorem~\ref{thm:mainresult} can be removed, following a similar approach to proving theorem~\ref{thm:mainresultTgates}, but we leave this for future work.

To prove theorem~\ref{thm:mainresult}, we show that gates of the form of equation~\eqref{Eq:QuditR} are transversal in all \mo{} codes for polynomial degree $r=\CD$.  To prove theorem~\ref{thm:mainresultTgates} we also prove that gates of the form in equation~\eqref{Eq:QuditT36} are transversal in $3^{\star}$-orthogonal codes. The theorems then follow by virtue of lemmas~\ref{rlemma1}, \ref{rlemma2} and~\ref{Colorcode_morth_lemma}, together with the known results for qubit color codes~\cite{Bombin07}. In particular, the transversal gates we consider are of star-conjugate transversal form, see equation~\ref{eqn:U_trans}.

\begin{lm}\label{mtranslemma}
A \mo{} code has a transversal implementation of the unitary gate $R_{f_{\CD}}$ such that:
\begin{equation}
\tilde{R}_{f_{\CD}} = R_{f_{\CD}}[v_{\bullet}] \otimes R_{f_{\CD}}^{*}[v_{\star}]
\label{eqn:R_trans}
\end{equation}
implements $R_{f_{\CD}}$ on the code space. 
\end{lm}

We prove this lemma in App.~\ref{sec:m-ortho}. As a simplified presentation for the main text of the paper, we prove the special and important case of $r=3$. The general proof follows the same approach, but is a little notationally unwieldy. 

\begin{proof}
\begin{lm}\label{tgatelemma}
A $3^{\star}$-orthogonal code has a transversal implementation of the unitary gate $T=R_{j_3}$ such that:
\begin{equation}
\tilde{T} = T[v_{\bullet}] \otimes T^{*}[v_{\star}]
\label{eqn:T_trans} 
\end{equation}
implements $T$ on the code space. 
\end{lm}
We now examine the conditions that the star-conjugate transversal $\tilde{T}$ implements a logical $\bar{T}$. For this to be true, the phases for each computational basis component of each logical codeword in equation~\eqref{eqn:log_state2}, must agree with the phases defining the gate in equation~\eqref{Eq:QuditT}. Noting that the same phase applies to all terms in the sum in equation~\eqref{eqn:log_state2}, we can write
\begin{equation}\label{Eq:starredT2}
\tilde{T} \ket{ (x,y)^{T}\cdot  G  } = \omega^{x^3} \ket{ (x,y)^{T}\cdot  G },
\end{equation}
where $x\in \mathbb{Z}_d$ and $y  \in \mathbb{Z}_d^{s}$, where $s$ is the number of stabilizer generators. In other words we fix the number of logical qubits $k=1$. 

Applying equations~\eqref{Eq:QuditT} and~\eqref{eqn:T_trans} to the state on the right hand side of equation~\eqref{Eq:starredT2}, we recover the following relationship between the phases on both sides of the equation:
\begin{align}
\label{GtoProve}
 &((x,y)^{T} \cdot G \cdot F) \circ ((x,y)^{T} \cdot G \cdot F) \circ \nonumber \\ 
  & \qquad  ((x,y)^{T} \cdot G \cdot F) = x^3 \pmod{\dd},
\end{align}
where we have right multiplied the $F$ matrix to vector $(x,y)^{T} \cdot G$. It is useful to define $z=(x,y)$ so
\begin{align}
  &\sum_{j}((z^{T} \cdot G \cdot F) \circ (z^{T} \cdot G \cdot F) \circ (z^{T} \cdot G \cdot F))_{j} \nonumber \\ 
  & \qquad\qquad\qquad\qquad= x^3 \pmod{\dd} .
  \label{eqn:exponents}
\end{align}  
Now define $z=\sum_{a} z_{a} e_{a} $ where $z_a$ is the value of the $a$th position of the vector $z$, and where $e_{a}$ are the basis vectors over $\mathbb{Z}_{2}^{s+1}$. For example $e_1 = (1,0,\ldots,0)$, $e_2 = (0, 1,\ldots,0)$ and so on. Re-writing the left hand side of equation~\eqref{eqn:exponents} we obtain
\begin{align}
  &\sum_{j}((z^{T} \cdot G \cdot F) \circ (z^{T} \cdot G \cdot F) \circ (z^{T} \cdot G \cdot F))_{j} \nonumber \\
  &= \sum_{j}\Bigg(\left(\sum_{a=1}^{s+1} z_{a} e^{T}_{a} \cdot G \cdot F \right) \circ \nonumber \\ 
  & \qquad \qquad\left(\sum_{b=1}^{s+1} z_{b} e^{T}_{b} \cdot G \cdot F\right) \circ \nonumber \\ 
  & \qquad \qquad\left(\sum_{c=1}^{s+1} z_{c} e^{T}_{c} \cdot G \cdot F\right)\Bigg)_{j} \pmod{\dd} . \nonumber
\end{align} 
Note that $e^{T}_{a} \cdot G$ is the $a$th row of the matrix G, which we denote $g_a$. Therefore 
\begin{align}
  & \sum_{j}((z^{T} \cdot G \cdot F) \circ (z^{T} \cdot G \cdot F) \circ (z^{T} \cdot G \cdot F))_{j} \nonumber \\
  &= \sum_{j}\Bigg(\left(\sum_{a=1}^{s+1} z_{a} g_{a} \cdot F\right) \circ \nonumber \\ 
  & \qquad \qquad \left(\sum_{b=1}^{s+1} z_{b} g_{b} \cdot F\right) \circ \nonumber \\ 
  & \qquad \qquad \left(\sum_{c=1}^{s+1} z_{c} g_{c} \cdot F\right)\Bigg)_{j} \pmod{\dd}. \nonumber
\end{align} 
we can re-write this as 
\begin{align}
  &\sum_{j}((z^{T} \cdot G \cdot F) \circ (z^{T} \cdot G \cdot F) \circ (z^{T} \cdot G \cdot F))_{j} \nonumber \\ 
  &= \sum_{j} \sum_{a,b,c=1}^{s+1} \big(\left( z_{a} g_{a} \cdot F\right) \circ \left( z_{b} g_{b} \cdot F\right) \circ \nonumber \\
  & \hspace{4cm} \left( z_{c} g_{c} \cdot F\right)\big)_{j} \pmod{\dd}, \nonumber \\ 
  &= \sum_{j}  \sum_{a,b,c=1}^{s+1} \left[ \left(g_{a} \circ g_{b} \circ g_{c} \right) \cdot F \right]_{j} z_{a} z_{b} z_{c} \pmod{\dd}.
   \label{eqn:ref_from_appendix}
\end{align} 
Summing over $j$ gives 
\begin{align}
  &\sum_{j}((z^{T} \cdot G \cdot F) \circ (z^{T} \cdot G \cdot F) \circ (z^{T} \cdot G \cdot F))_{j}\nonumber \\
  & = \sum_{a,b,c=1}^{s+1} \left| \left(g_{a} \circ g_{b} \circ g_{c} \right) \cdot F \right| z_{a} z_{b} z_{c} \pmod{\dd}.
\end{align} 
We break the sum into two parts, so that
\begin{align}
\label{EqnLongG}
  & \sum_{j}((z^{T} \cdot G \cdot F) \circ (z^{T} \cdot G \cdot F) \circ (z^{T} \cdot G \cdot F))_{j} \nonumber \\  
   &=   \sum_{a}  | (g_{a} \circ g_{a} \circ g_{a})\cdot F | z_{a}^3  + \nonumber \\
   &  \sum_{\{ a, b, c | \neg (a = b =c) \}  }  |(g_{a} \circ g_{b} \circ g_{c})\cdot F | z_{a} z_{b } z_{c} \nonumber \\ &\hspace{5cm} \pmod{\dd}.
\end{align}
Now we use the definition of \mo{} matrices. This ensures that the first term is equal to $x^3\pmod {d}$ and the second term is equal to zero. Hence equation~\eqref{GtoProve} is satisfied, which completes the proof of lemma~\ref{tgatelemma}.  
\end{proof}

For larger $\CD$, the proof follows the same lines but is with more cumbersome algebra. The proof of lemma~\ref{mtranslemma} which we present in App.~\ref{sec:m-ortho} completes the proof of theorem~\ref{thm:mainresult}. To prove theorem~\ref{thm:mainresultTgates} we require one further lemma:
 
\begin{lm}\label{tgatelemma36}
A $3^{\star}$-orthogonal qudit code where $d=3$ or $d=6$ has a transversal implementation of the unitary gate $T_{3,6}$ defined in equation~\eqref{Eq:QuditT36} such that:
\begin{equation}\label{eqn:T36_trans} 
\tilde{T_{3,6}} = T_{3,6}[v_{\bullet}] \otimes T_{3,6}^{*}[v_{\star}]
\end{equation}
\end{lm}
 
The proof is presented in App.~\ref{Sec:triortho_3_and_6}, completing the proof of theorem~\ref{thm:mainresultTgates}. Finally, we remark, that as CSS codes, all color codes admit a transversal $\Lambda(X)$. We conclude this section with a table of star-conjugate transversal gates in qudit color codes of different spatial dimension, see table~\ref{Gatetable}.

\begin{table*}

\begin{tabularx}{\textwidth}{ | c |c | Y | c|}
\hline
Logical operator & Defined in Equation & Star-conjugate transversal implementation & Spatial dimension \\
\hline
$H$ & equation~\eqref{Eq:QuditH} &$H[v_{\bullet}] \otimes H^{*}[v_{\star}]$ & $2$ \\
$\Lambda(X)$ & equation~\eqref{Eq:QuditCX} & Applied blockwise transversally & $\ge 2$ \\
$S$ & equation~\eqref{Eq:QuditS} &$S[v_{\bullet}] \otimes S^{*}[v_{\star}]$ & $\ge 2$ \\
$T$ & equation~\eqref{Eq:QuditT} &$T[v_{\bullet}] \otimes T^{*}[v_{\star}]$ & $\ge 3$ \\
$R_{f_{r}}$ & equation~\eqref{Eq:QuditR} &$R_{f_{r}}[v_{\bullet}] \otimes R_{f_{r}}^{*}[v_{\star}]$ & $\ge r$ \\
\hline
	
\end{tabularx}

\caption{\label{Gatetable} A table illustrating the star-conjugate transversal implementation of some important logical gates and the spatial dimensions of the codes which support them. }	
\end{table*}

\section{Gauge fixing to implement the logical Hadamard}\label{sec:gauge_fixing}

In this section and Sec.~\ref{sec:gauge_color_codes} we show how a universal gate set can be achieved via gauge fixing, or by defining a gauge color code. Although these results follow directly from the qubit cases presented in~\cite{Paetznick13} and~\cite{Bombin13} without the need for any further technical innovations, we include a discussion and explanation of these techniques for completeness.

First we consider gauge fixing. We shall define a \textit{subsystem code} occupying the same lattice as the color code. This code is capable of realising the logical Hadamard gate on the same encoded qudit as the color code, while a set of gauge qudits become corrupted. Fortunately there is a simple protocol to fix the corrupted gauge qudits to re-initialize the color code.

In a $[[n,k,d]]$ subsystem code~\cite{Bacon06,Kribs05,Poulin05,Kribs06} there are three sets of operators that define the code. First, a set of gauge generators is defined, a set of Pauli operators defined on the lattice that do not necessarily mutually commute. These generate the gauge group $\mathcal{A}$. The center of the gauge group (the elements of the group which commute with all other elements) represents the stabilizer of the code $\mathcal{S}$ (up to a free choice of plus or minus phases which must be chosen to define the zero-error code space). The final set of operators to define is the set of logical Pauli operators on the code space, the Pauli operators that commute with $\mathcal{A}$. Subsystem codes can be understood as stabilizer codes in which some of the logical qudits have been demoted to gauge qubits. These gauge qubits are not used to carry information and can be corrupted or measured during operations on the code. 

The logical operators $\bar{X}$ and $\bar{Z}$ are products of $X$ and $Z$ (and $X^*$ and $Z^*$) operators acting on every qudit.  Therefore, a natural candidate for transversal $\bar{H}$ is a global star-conjugate Hadamard.  However, for color codes of $\DD>2$ it can be seen from the structure of the stabilizer generators that such a unitary will not leave the code space invariant. In particular, the $X$ stabilizers act on the $\DD^{\prime}$-cells of the lattice whereas the $Z$ stabilizers act on the $(\DD-\DD^{\prime}+2)$-cells of the lattice. Therefore unless $\DD^{\prime} = \DD-\DD^{\prime}+2$, the $Z$ stabilizers cannot be transformed into valid $X$ (cell) stabilizers, by a global Hadamard. The solution to this is to switch between the stabilizer color codes and a different code in which $X$ stabilizer generators and $Z$ stabilizer generators have identical support.

In $\DD$ spatial dimensions the subsystem code stabilizers are defined as the star-conjugate $X$ and $Z$ operators acting on the qudits contained in a $\DD$-cell of the lattice. The logical operators matching the color code logical operators are the transversal $\bar{X}$ and $\bar{Z}$ defined in equations~\eqref{eqn:log_ops}. In addition there are $X$ and $Z$ gauge operators defined on the $(\DD-1)$-cells, $(\DD-2)$-cells, and so on to the $2$-cells of the lattice, such that they commute with the stabilizer group. They can be identified in $\omega$-commuting pairs, with each pair protecting a gauge qudit.

In higher spatial dimensions there are a larger number of gauge generators and hence more gauge qudits than in lower spatial dimensions. Nevertheless, the basic protocol for realising the logical Hadamard gate in the subsystem code and fixing the gauge to return to the stabilizer code is independent of the spatial dimensions. 

The star-conjugated transversal Hadamard gate,
\begin{eqnarray}
\bar{H} = \tilde{H} = H[v_{\bullet}] \otimes H^{*}[v_{\star}],
\label{eqn:trans_H}
\end{eqnarray}
transforms the transversal logical operators and the stabilizers of the code correctly. However the gauge generators are corrupted under the action of $\bar{H}$, meaning that $\bar{Z}_{j}$ (the $Z$ logical operator for qudit $j$) is transformed not to $\bar{X}_{j}$ as desired but instead to the operator for a different qudit, $\bar{X}_{k}$, causing the gauge qudits to become entangled. 

These corrupted gauge operators must be \textit{fixed}, to return the code to the original stabilizer code in which the other transversal operators can be realized. This fixing is enacted by measuring the stabilizer generators. The measured eigenvalues will indicate the Pauli correction operator that should be applied---this is an operator that restores the corrupted gauge qudit to its initial state while commuting with the other logical operators.

Once the corrections are applied, the gauge qudits are re-initialized to their original state so that any entanglement between the gauge qudits is destroyed. This is equivalent to the stabilizer color code with a logical Hadamard applied to the encoded qudit.



\section{Qudit gauge color codes} \label{sec:gauge_color_codes}

Another approach to the implementation of the transversal Hadamard is to treat the subsystem code  as the base code, and to achieve the transversal non-Clifford gates  by gauge fixing to the  (stabilizer) color code. This switch in perspective leads to the \textit{gauge color codes}, recently introduced by Bombin~\cite{Bombin13}. Far more than a change in perspective, however, gauge color codes gain many useful new properties, which the original color codes do not possess.

The first advantage of the gauge color code construction is that the outcome of stabilizer measurements can  be inferred from measurements of the gauge operators. This offers a practical advantage, since the weight of the gauge generators can be significantly smaller than stabilizer generators. For example, in 3D the $X$ stabilizer generators would include 24-body operators associated with truncated octahedrons, discussed in Sec.~\ref{sec:code_distance}, requiring coherent 24-body measurement. In the equivalent gauge color code, however, the $X$ stabilizers can be decomposed into 4- or 6-body $X$ gauge operator measurements, a much more practically feasible proposition. 

A second advantage (in the $\DD=3$ case) is that the outcome of stabilizer measurements is more robustly encoded in the gauge generator measurement. The extra redundancy leads to to ``single-shot fault-tolerant error correction''~\cite{Bombin14singleshot}---meaning that measurement errors can be accounted for without the need for repeated measurements.

A third advantage is that one can use gauge fixing to fault tolerantly convert~\cite{Anderson14} between codes of different spatial dimension~\cite{Bombin14dimension}, allowing one to combine the advantageous features of both. For example, one could perform all Clifford computations in a 2-D code while only resorting to a 3-D code to achieve a non-Clifford gate.

 We can define qudit gauge color codes in any spatial dimension analogously to Bombin's qubit codes~\cite{Bombin13}. The only difference is that operators are star-conjugated according to the star-bipartition of the lattice.

The most studied gauge color code exists on a 3D lattice and we shall describe the qudit analogue as an example. Gauge generators consist of star-conjugated face operators of both $X$ and $Z$-type. The stabilizer is generated by  3-cell operators whose values can be obtained from the products of the face operators. The star-conjugate transversality of the Hadamard gate for this code is trivial via its definition. To achieve a transversal $T$ gate one must gauge fix back to the stabilizer color code by promoting $Z$ face operators to the stabilizer.

Although the basic properties of color codes transfer immediately to the general qudit case via this definition and by the results presented in previous sections in this paper, the performance of single-shot error correction is less straightforward to analyze, and we leave a detailed study to further work. Single-shot error correction works, in the qubit case, because the relationship between gauge and stabilizer measurement outcomes in the code is such that the values of the stabilizer are mapped to delocalized structures (branching points in string-like extended structures) in the gauge outcomes. This arises because the parity of the outcomes of the face operators around the cell must be equal to the eigenvalue of the stabilizer representing the cell, even parity in for a trivial syndrome, and odd-parity for the non-trivial case.

In qudit codes, the stabilizer measurement outcome may take any of $\dd-1$ non-trivial values. The outcomes of gauge measurements must sum to these values. Thus the relationship between gauge measurements and stabilizer measurements is more complicated than the simple ``branching point or no branching point'' behavior seen in the qubit code. This extra structure, however, should represent additional information about the location and charge of stabilizer outcomes which a decoder could exploit. Thus, with the development of a suitable decoder, we expect  single-shot error correction robust to measurement error to be achieved in qudit gauge color codes, but we leave its the construction of such a decoder, and analysis of its performance, to future work.


\begin{figure}
\includegraphics[width=0.3\textwidth]{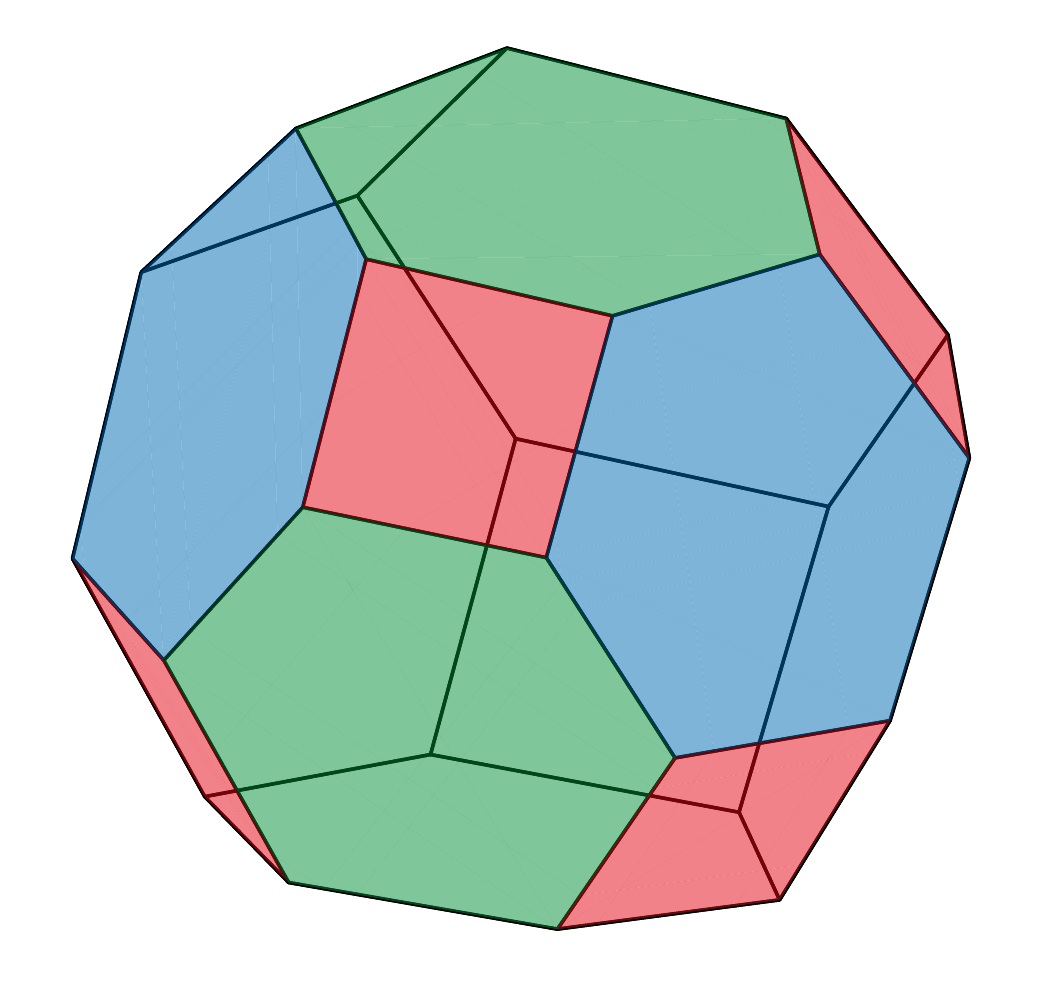} \centering
\caption{(Color online) A cell of the 3D gauge color code. Its faces are 3-colorable and colored red, green and blue in this example. The stabilizer outcomes for the cell can be constructed by adding together $\pmod{\dd}$ the outcomes of the gauge measurements for each of the red faces, or each of the blue faces, or each of the green faces.}
\label{fig:gauge_faces}
\end{figure}



\section{Discussion} \label{sec:conclusion}

We have shown how any existing qubit color code can be generalized to support a qudit color code in arbitrary spatial dimensions, and with an arbitrarily large distance. To do this we introduced the notion of a star-conjugate lattice and showed how this construction allows the Clifford phase, controlled-$X$ and non-Clifford phase gates to be implemented transversally. The set of Clifford group generators is completed by the Hadamard, which can be implemented in a subsystem code. Techniques for switching between the two codes have been outlined. An additional advantage of considering the color codes as gauge codes is reduction of the weight of measurements required to obtain the error syndrome.

There are still many open questions remaining---here we shall outline some of the ones we find most interesting. Although a fault-tolerant decoder has been implemented for a 2D code, no decoder has, to our knowledge, been proposed for color codes in higher spatial dimensions. In Ref.~\cite{Bombin14singleshot} a decoder for a 3D gauge color code is discussed, unfortunately however, a classical algorithm capable of performing the decoding is not specified. One would expect that in a fully fault-tolerant simulation the gate error threshold for the $\DD$-dimensional color codes would be prohibitively low, since the stabilizers defined on the $\DD$-cells would require circuits of many time steps to measure. There is hope however that by switching to the corresponding gauge color code the measurements require fewer gates and therefore the threshold may recover~\cite{Brown15}. The development of efficient decoders for qudit gauge color codes is therefore of significant interest.

The construction of the transversal gates in the color codes and gauge color codes presented in this work holds for qudits of any dimension, $\dd$. However, for non-prime $\dd$ it is not clear that $S$, $H$ and $\Lambda(X)$ are sufficient to generate the Clifford group.  We desire a better understanding of universality and the Clifford hierarchy when using systems of non-prime dimension.

We have identified some other lattices that can support a color code in 3D. It would be interesting to see a comprehensive list of possible lattices. Furthermore, our star-conjugate construction is valid for all of these lattices, but there may be one that offers advantages in terms of the necessary qudit overhead to achieve a desired distance.  

Finally, there are many Reed-Muller codes known for qudits for which no equivalent color code is known~\cite{Campbell12,Campbell14}. Is it possible to use one of these to construct a topological code similar to a color code and study its properties?  Here, we found codes where the distance is purely topological in nature, and so depends on the lattice not on whether we use qudits.  In contrast, qudit Reed-Muller codes provide an improved distance that is algebraic rather than topological in origin.  The potential exists for qudit quantum codes which abstract topological ideas to a more general  setting to generate novel codes of arbitrary code distance with a rich family of transversal logical gates.

\subsection*{Acknowledgements}

DEB would like to thank Hector Bombin and Ben Brown for helpful and elucidatory discussions on color codes and gauge color codes.
FHEW and HA acknowledge the financial support of the EPSRC (grant numbers  EP/G037043/1 and EP/K022512/1, respectively).




\appendix

\section{Proof of lemma~\ref{mtranslemma}}\label{sec:m-ortho}

\setcounter{lm}{6}
\begin{lm}
A \mo{} code has a transversal implementation of the unitary gate $R_{f_{\CD}}$ such that:
\begin{equation}
\tilde{R}_{f_{\CD}} = R_{f_{\CD}}[v_{\bullet}] \otimes R_{f_{\CD}}^{*}[v_{\star}]
\label{appendixlemma7}
\end{equation}
implements $R_{f_{\CD}}$ on the code space. 
\end{lm}

\begin{proof}
For equation~\eqref{appendixlemma7} to be true, for all $z=(x,y)$ we have
\begin{equation}
    \tilde{R}_{f_{\CD}} \ket{ z^{T}\cdot G  } = \omega^{  z_{j}^{\CD}} \ket{ z^{T}\cdot G }.
    \label{Eq:starredR3}
\end{equation}
Applying equations~\eqref{Eq:QuditR} and~\eqref{eqn:R_trans} to the state on the right hand side of equation~\eqref{Eq:starredR3}, the phases on both sides of the equation are equal if
\begin{equation}
 (z^{T} \cdot G \cdot F)^{\circ \CD} = x^{\CD} \pmod{\dd},
   \label{eqn:m_exponents}
\end{equation}
where the notation $(.)^{\circ \CD}$ means a circle-product (see equation~\eqref{eqn:circle_prod}) taken between $\CD$ identical elements. 
Proceeding using the same notation as in Sec.~\ref{sec:transversal} we find 
\begin{align}
  &\sum_{j}((z^{T} \cdot G \cdot F)^{\circ \CD})_{j} \nonumber \\
  &= \sum_{j}\left(\left( \sum_{a} z_{a} g_{a}.F \right)^{ \circ \CD }\right)_{j} \pmod{\dd} , \nonumber \\
  &= \sum_{j}\left( \sum_{\vec{a} \in [1,\ldots,k+s]^{\CD}} z_{\vec{a}} \: g_{\vec{a}}.F \right)_{j} \pmod{\dd},
\end{align} 
where $[a,\ldots,b]^{c}$ is a vector of length $c$ with elements that take values between $a$ and $b$, and where $g_{\vec{a}}$ is the bitwise product of $\CD$ rows of $G$, so
\begin{equation}
    g_{\vec{a}} = g_{a_1} \circ g_{a_2} \circ \ldots  \circ g_{a_{\CD}}.
\end{equation}
We can re-write this as 
\begin{align}
  &\sum_{j}((z^{T} \cdot G \cdot F)^{\circ \CD})_{j} \nonumber \\ 
  & = \sum_{j}  \sum_{\vec{a} \in [1,\ldots,k+s]^{\CD}} \left[ g_{\vec{a}} \cdot F \right]_{j} z_{\vec{a}} \pmod{\dd}.
\end{align} 
Summing over $j$ gives 
\begin{align}
  &\sum_{j}((z^{T} \cdot G \cdot F)^{\circ \CD})_{j} \nonumber \\
  & = \sum_{\vec{a} \in [1,\ldots,k+s]^{\CD}} \left| g_{\vec{a}} \cdot F \right| z_{\vec{a}} \pmod{\dd}.
   \label{eqn:long_m-ortho}
\end{align} 

We now use the definition of a \mo{} code. This means that all terms satisfying $\neg (a_{1} = \ldots = a_{\CD})$ in the summation on the right hand side vanish. This implies that the only non-zero contribution is represented by $\vec{a}= (a,a,\ldots,a)$, a $\CD$-element vector where every element is $a$. Recall also that such terms vanish unless $a=1$, in other words it is the $\CD$-fold elementwise product of the vector in $G_{1}$. We note that $z_{\vec{a}} = z_{a}^{\CD}$, leaving only
\begin{align}
   \sum_{j}((z^{T} \cdot G \cdot F)^{\circ \CD})_{j} &=  | g_{\vec{a}} \cdot F | z_{a}^{\CD} \pmod{\dd}, \nonumber\\
   &=  x^{\CD} \pmod{\dd},
\end{align}
which proves the desired equation~\eqref{eqn:m_exponents} is satisfied.  
\end{proof}

\section{Proof of lemma~\ref{tgatelemma36}} \label{Sec:triortho_3_and_6}

\setcounter{lm}{8}
\begin{lm}
A $3^{\star}$-orthogonal qudit code where $d=3$ or $d=6$ has a transversal implementation of the unitary gate $T_{3,6}$ defined in equation~\eqref{Eq:QuditT36} such that:
\begin{equation} 
\tilde{T}_{3,6} = T_{3,6}[v_{\bullet}] \otimes T_{3,6}^{*}[v_{\star}]
\end{equation}
\end{lm}

where 

\begin{equation}
T_{3,6} = \sum_{j \in \mathbb{Z}_{\dd}} \gamma^{j^3} \ket{j}\bra{j}
\label{Eq:QuditT36app}
\end{equation}
where $\gamma^3=\omega$ and where the function in the exponent $j^3$ is evaluated in regular arithmetic (or equivalently modulo $3d$). We shall  prove lemma~\ref{tgatelemma36} for $\dd=3$ and $\dd=6$ separately, since each requires calculation in arithmetic with a different modularity.

%

\subsection{$\dd=3$}

\begin{proof}
First we consider the case of qutrits, where $d=3$. In this case,   $\gamma=e^{i \frac{2\pi}{9}}$. Hence, our calculation will make use of addition modulo 9. Proceeding as before we start with:
\begin{equation}
    \tilde{T}_{3,6} \ket{ z^{T}\cdot G  } = \gamma^{ x^{3}} \ket{ z^{T}\cdot G }.
\end{equation}
The phases on both sides of this equation match if
\begin{equation}
  \sum_{j}(( z^{T}\cdot  G \cdot F)^{\circ 3} )_{j} =  x^3 \pmod{9}.
  \label{eqn:d3_m_eq}
\end{equation}  
This expression is a mix of modulo 3 and modulo 9 arithmetic so we must proceed with care. The matrix multiplication is mod 3, whereas all other arithmetic is mod 9. Re-writing the left hand side as
\begin{align}
	&\sum_{j}(( z^{T}\cdot  G \cdot F)^{\circ 3} )_{j} \\ \nonumber &= \sum_{j} \sum_{a,b,c=1}^{s+1} \left\{ [g_{a} \cdot F]_{j} z_{a}  \pmod{3} \right\} \circ \nonumber \\
	& \quad \left\{ [g_{b} \cdot F]_{j} z_{b}  \pmod{3} \right\} \circ \nonumber \\
	& \qquad \left\{ [g_{c} \cdot F]_{j} z_{c}  \pmod{3} \right\},
\end{align}  
 we can use the following identity to express the modulo 3 reduction in terms of standard addition 
\begin{equation}
a \pmod{n} = a - \lfloor \frac{a}{n} \rfloor n.
\label{eqn:mod_def}
\end{equation}
Setting aside the modulo 9 arithmetic for the moment and evaluating this expression first in standard arithmetic,  the left hand side now expands to the following form:
\begin{widetext}
\begin{equation}
\begin{split}	
\sum_{j}(( z^{T}\cdot  G \cdot F)^{\circ 3} )_{j}	&= \sum_{j} \sum_{a,b,c=1}^{s+1} \left( [g_{a} \cdot F]_{j} z_{a} - 3\lfloor \frac{[g_{a} \cdot F]_{j} z_{a}}{3} \rfloor \right) \nonumber \\
	& \qquad \circ \left( [g_{b} \cdot F]_{j} z_{b} - 3\lfloor \frac{[g_{b} \cdot F]_{j} z_{b}}{3} \rfloor \right) \circ \left( [g_{c} \cdot F]_{j} z_{c} - 3\lfloor \frac{[g_{c} \cdot F]_{j} z_{c}}{3} \rfloor \right)  \\
	&= \sum_{j} \sum_{a,b,c=1}^{s+1} [(g_{a} \circ g_{b} \circ g_{c}) \cdot F]_{j} z_{a} z_{b} z_{c} - 9 [(g_{a} \circ g_{b})\cdot F]_{j} z_{a} z_{b} \lfloor \frac{[g_{c} \cdot F]_{j} z_{c}}{3} \rfloor \nonumber\\ & \qquad + 27 [g_{a} \cdot F]_{j} z_{a} \lfloor \frac{[g_{b} \cdot F]_{j} z_{b}}{3}\rfloor \lfloor \frac{[g_{c} \cdot F]_{j} z_{c}}{3}\rfloor \nonumber\\ & \qquad\qquad  - 27 \lfloor \frac{[g_{a}\cdot F]_{j} z_{a}}{3}\rfloor\lfloor \frac{[g_{b} \cdot F]_{j} z_{b}}{3}\rfloor \lfloor \frac{[g_{c}\cdot F]_{j} z_{c}}{3}\rfloor.
\end{split}
\end{equation}
\end{widetext}

The last three terms in this summation are integer multiples of 9 and are thus equal $0 \pmod{9}$. Thus if we reimpose modulo 9 arithmetic we are left with:
\begin{align}
&\sum_{j}(( z^{T}\cdot  G \cdot F)^{\circ 3} )_{j}	\pmod{9}\nonumber \\
&=\sum_{j} \sum_{a,b,c=1}^{s+1} [(g_{a} \circ g_{b} \circ g_{c}) \cdot F]_{j} z_{a} z_{b} z_{c} \nonumber \\
& \hspace{4cm} \pmod{9}
\end{align}

Finally, invoking \moity{} completes the proof. 
\end{proof}

Note that we are using the fact that the definition of \moity is stated in regular arithmetic which implies the modulo 9 orthogonality used here.

\subsection{$\dd=6$}

\begin{proof}
When $d=6$ the phase in $T_{3,6}$ is $\gamma=e^{i \frac{2\pi}{18}}$ and thus proceed with modulo 18 arithmetic. 
\begin{equation}
    \tilde{T}_{3,6} \ket{ z^{T}\cdot G  } = \gamma^{ x^{3}} \ket{ z^{T}\cdot G }.
\end{equation}
We expand the right hand side and find that
\begin{align}
	&\sum_{j}(( z^{T}\cdot  G \cdot F)^{\circ 3} )_{j} \nonumber \\
	&= \sum_{j} \sum_{a,b,c=1}^{s+1} [\left(g_{a} \circ g_{b} \circ g_{c} \right)\cdot F]_{j} z_{a} z_{b} z_{c} \pmod{18}.
\end{align}

Converting modulo 18 arithmetic to normal arithmetic we find,
\begin{widetext}
\begin{align}
\begin{split}
	\sum_{j}(( z^{T}\cdot  G \cdot F)^{\circ 3} )_{j} &= \sum_{j} \sum_{a,b,c=1}^{s+1} [(g_{a} \circ g_{b} \circ g_{c}) \cdot F]_{j} z_{a} z_{b} z_{c} - 18 [(g_{a} \circ g_{b})\cdot F]_{j} z_{a} z_{b} \lfloor \frac{[g_{c} \cdot F]_{j} z_{c}}{6} \rfloor \nonumber\\ & \qquad + 108 [g_{a} \cdot F]_{j} z_{a} \lfloor \frac{[g_{b} \cdot F]_{j} z_{b}}{6}\rfloor \lfloor \frac{[g_{c} \cdot F]_{j} z_{c}}{6}\rfloor \nonumber\\ & \qquad\qquad - 216 \lfloor \frac{[g_{a}\cdot F]_{j} z_{a}}{6}\rfloor\lfloor \frac{[g_{b} \cdot F]_{j} z_{b}}{6}\rfloor \lfloor \frac{[g_{c}\cdot F]_{j} z_{c}}{6}\rfloor .
\end{split}
\end{align}
\end{widetext}
Casting this back into modulo 18 arithmetic we can cancel the terms equal to integer multiples of 18, leaving:
\begin{align}
&\sum_{j}(( z^{T}\cdot  G \cdot F)^{\circ 3} )_{j} \pmod{18} \nonumber\\ 
&= \sum_{j} \sum_{a,b,c=1}^{s+1} [(g_{a} \circ g_{b} \circ g_{c}) \cdot F]_{j} z_{a} z_{b} z_{c}	\pmod{18}
\end{align}

Finally, via \moity{} we recover
\begin{equation}
\sum_{j}(( z^{T}\cdot  G \cdot F)^{\circ 3} )_{j} \pmod{18}= x^3	\pmod{18}
\end{equation}
completing the proof. 
\end{proof}

\clearpage

\end{document}